\documentclass[10pt,twocolumn]{article}%
\usepackage{latex8}
\usepackage{amsmath}
\usepackage{amssymb}
\usepackage{graphicx}%
\providecommand{\U}[1]{\protect\rule{.1in}{.1in}}
\newenvironment{proof}[1][Proof]{\noindent\textbf{#1.} }{\ \rule{0.5em}{0.5em}}
\begin{document}

\title{Quantum Copy-Protection and Quantum Money}
\author{Scott Aaronson\thanks{MIT. \ Email: aaronson@csail.mit.edu. \ Partly supported
by the Keck Foundation. \ Some of this work was done while the author was at
the University of Waterloo.}}
\date{}
\maketitle

\begin{abstract}
Forty years ago, Wiesner proposed using quantum states to create money that is
physically impossible to counterfeit, something that cannot be done in the
classical world. \ However, Wiesner's scheme required a central bank to verify
the money, and the question of whether there can be unclonable quantum money
that anyone can verify has remained open since. \ One can also ask a related
question, which seems to be new: can quantum states be used as
\textit{copy-protected programs}, which let the user evaluate some function
$f$, but not create more programs for $f$?

This paper tackles both questions using the arsenal of modern computational
complexity. \ Our main result is that there exist \textit{quantum
oracles}\ relative to which publicly-verifiable quantum money is possible, and
any family of functions that cannot be efficiently learned from its
input-output behavior can be quantumly copy-protected. \ This provides the
first formal evidence that these tasks are achievable. \ The technical core of
our result is a \textquotedblleft Complexity-Theoretic No-Cloning
Theorem,\textquotedblright\ which generalizes both the standard No-Cloning
Theorem and the optimality of Grover search, and might be of independent
interest. \ Our security argument also requires explicit constructions of
quantum $t$-designs.

Moving beyond the oracle world, we also present an \textit{explicit candidate
scheme} for publicly-verifiable quantum money, based on random stabilizer
states; as well as two explicit schemes for copy-protecting the family of
point functions. \ We do not know how to base the security of these schemes on
any existing cryptographic assumption. \ (Note that without an oracle, we can
only hope for security under \textit{some} computational assumption.)

\end{abstract}

\section{Introduction\label{INTRO}}

In classical physics, any information that can be read can be copied an
unlimited number of times---which is why the makers of software, music CDs,
and so on have met such severe difficulties enforcing \textquotedblleft
digital rights management\textquotedblright\ on their products (see Halderman
\cite{halderman}\ for example). \ Quantum states, on the other hand, cannot in
general be copied, since measurement is an irreversible process that destroys
coherence. \ And this immediately raises the possibility of using quantum
states as \textit{unclonable information}, such as money or copy-protected
software. \ The idea of using quantum states in this way actually predates
quantum information as a field. \ In a remarkable 1970 manuscript that first
discussed the idea of quantum cryptography, Wiesner \cite{wiesner} also
proposed a scheme for \textquotedblleft quantum money\textquotedblright\ that
a central bank can prepare and verify, but that is information-theoretically
impossible for anyone other than the bank to copy. \ Wiesner's result
immediately raised a question: could there be quantum money states that
\textit{anyone} can verify---that is, for which the authentication procedure
is completely public---but that are still infeasible to copy? \ This latter
question has remained open for forty years.

However, while the quantum money problem is fascinating by itself, it also
motivates a broader problem: \textit{what sorts of \textquotedblleft
unclonable power\textquotedblright\ can be provided by a quantum state?} \ So
for example, given a Boolean function $f:\left\{  0,1\right\}  ^{n}%
\rightarrow\left\{  0,1\right\}  $, one can ask: does there exist a quantum
state $\left\vert \psi_{f}\right\rangle $ that lets its owner compute $f$ in
polynomial time, but does \textit{not} let her efficiently prepare more states
that are also useful for computing $f$?\footnote{Formally, of course, we would
want a scheme that worked for a whole family of $f$'s.} \ Such a state could
be interpreted as \textquotedblleft quantumly copy-protected
software.\textquotedblright\ \ Whereas in the quantum money problem, we wanted
unclonable states that could be \textit{verified as authentic}, in the quantum
copy-protection problem we want unclonable states that let us \textit{do
something useful} (namely, compute $f$). \ There are other interesting
unclonable functionalities (such as identity cards), but in this paper, money
and copy-protected software will be more than enough to occupy us.

A first, crucial observation is that, if we insist on
\textit{information-theoretic security} (as provided, for example, by quantum
key distribution), then we cannot hope for either quantum copy-protection
\textit{or} publicly-verifiable quantum money. \ The reason is simple: an
adversary with unlimited computational power could loop through all possible
quantum states $\left\vert \psi\right\rangle $, halting only when it found a
state\ with the required properties.\footnote{In the copy-protection case, the
property of $\left\vert \psi\right\rangle $\ that we care about is that of
\textquotedblleft being a valid quantum program for the function
$f$.\textquotedblright\ \ And even if an explicit description of $f$ is not
available, this property can be checked using unlimited computational power,
together with polynomially many copies of $\left\vert \psi_{f}\right\rangle $
(the \textquotedblleft store-bought\textquotedblright\ quantum program for
$f$).} \ Therefore, if we want these functionalities, we are going to have to
make computational hardness assumptions. \ However, this does not by any means
defeat the purpose---for remember that in the classical world, the
functionalities we are talking about are flat-out impossible, regardless of
what computational assumptions we make.

In our view, unclonable information remains one of the most striking potential
applications of quantum mechanics to computer science. \ Firstly, unclonable
information would solve problems of clear, longstanding, and undisputed
importance in the classical world---in the case of money that cannot be
counterfeited, a problem that people have been trying to solve for thousands
of years. \ Secondly, unlike with (say) quantum cryptography, the problems
being addressed here are ones for which theoretically-grounded classical
alternatives simply \textit{do not exist}, because of the copyability of
classical information.\footnote{Here we are leaving aside solutions that
involve repeated interaction with a server: we seek solutions in which the
cash, software, etc. can be placed under the complete physical control of the
user.} \ Thirdly, as we will see, some quantum money proposals require no
multi-qubit entanglement, and might therefore be technologically feasible long
before general-purpose quantum computing.\footnote{On the other hand, quantum
money must be protected from decohering, and this remains the central
technological obstacle to realizing it. \ Depending on the physical substrate,
right now qubits can be stored coherently for a few seconds or at most
minutes, and only in laboratory conditions. \ By contrast, quantum key
distribution (QKD) requires only the \textit{transmission} of coherent qubits
and not their long-term \textit{storage}---which is why QKD can be implemented
even with today's technology.}

Given all this, it is surprising that the questions of unclonable quantum
money and software have barely been taken up over the last two decades. \ The
goal of this paper is to revisit these questions using the arsenal of modern
theoretical computer science. \ Our main result (stated informally) is the following:

\begin{theorem}
\label{main}There exists a quantum oracle $U$\ relative to which
publicly-verifiable quantum money and quantum copy-protection of
arbitrary\textit{\footnote{By which we mean, software that is not learnable
from its input/output behavior using a polynomial-time quantum computation.
\ Learnable software is impossible to copy-protect for trivial reasons. \ Our
result shows that relative to an oracle, this is the only obstruction.}}
software are possible.\footnote{For technical reasons, the copy-protection
result currently only gives security against pirating algorithms that
\textit{more than double} the number of programs. \ We hope to remove this
restriction in the future.}
\end{theorem}

Here a \textquotedblleft quantum oracle,\textquotedblright\ as defined by
Aaronson and Kuperberg \cite{ak}, is just an infinite collection of unitary
operations $U=\left\{  U_{n}\right\}  _{n\geq1}$\ that can be applied in a
black-box fashion. \ Theorem \ref{main} implies that, if quantum money and
copy-protection are \textit{not} possible, then any proof of that fact will
require \textquotedblleft quantumly non-relativizing
techniques\textquotedblright: techniques that are sensitive to the presence of
a quantum oracle. \ Such a proof is almost certainly beyond present-day techniques.

However, we also go beyond oracle results, and propose \textit{the first
explicit candidate schemes for publicly-verifiable quantum money and for
copy-protecting the family of point functions}.\ \ Here a \textquotedblleft
point function\textquotedblright\ is a Boolean function $f_{s}:\left\{
0,1\right\}  ^{n}\rightarrow\left\{  0,1\right\}  $\ such that $f_{s}\left(
x\right)  =1$\ if and only if $x$ equals some secret string $s$.
\ Copy-protecting point functions has an interesting application for computer
security: it yields a way to distribute a password-authentication
program\ such that, not only can one not learn the password by examining the
program, one cannot even use the program to create additional programs with
the ability to \textit{recognize} the password.

Our candidate quantum money scheme is based on random stabilizer states; the
problem of counterfeiting the money is closely related to noisy decoding for
random linear codes over $\mathbb{GF}_{2}$. \ For copy-protecting point
functions, we actually give \textit{two} schemes: one based on random quantum
circuits (as recently studied by Harrow and Low \cite{harrowlow}), the other
based on hidden subgroups of the symmetric group. \ The key challenge, which
we leave unresolved, is to base the security of our schemes on a
\textquotedblleft standard\textquotedblright\ cryptographic assumption (for
example, the existence of pseudorandom functions secure against quantum
attack), as opposed to the tautological assumption that our schemes are secure!

Our results give the first complexity-theoretic evidence that quantum
copy-protection and publicly-verifiable quantum money are indeed possible.
\ On the other hand, the oracle results also help explain the
\textit{difficulty} of proving explicit schemes for these tasks secure. \ For
as we will see, proving security in the oracle world is already highly
nontrivial! \ Furthermore, \textit{any security proof for an explicit scheme
will need to include our oracle result as a special case}---since an attack on
an explicit scheme could always proceed by treating all the relevant circuits
as black boxes and ignoring their internal structure.

\subsection{Techniques\label{TECHNIQUES}}

In proving our oracle results, perhaps the most novel technical ingredient is
what we call the \textquotedblleft Complexity-Theoretic No-Cloning
Theorem\textquotedblright:

\begin{theorem}
[Complexity-Theoretic No-Cloning]\label{ctnc0}Let $\left\vert \psi
\right\rangle $\ be an $n$-qubit pure state. \ Suppose we are given the
initial state $\left\vert \psi\right\rangle ^{\otimes k}$ for some $k\geq1$,
as well as an oracle $U_{\psi}$\ such that $U_{\psi}\left\vert \psi
\right\rangle =-\left\vert \psi\right\rangle $\ and $U_{\psi}\left\vert
\phi\right\rangle =\left\vert \phi\right\rangle $\ for all $\left\vert
\phi\right\rangle $ orthogonal to $\left\vert \psi\right\rangle $. $\ $Then
for all $\ell>k$, to prepare $\ell$\ registers $\rho_{1},\ldots,\rho_{\ell}%
$\ such that%
\[
\sum_{i=1}^{\ell}\left\langle \psi\right\vert \rho_{i}\left\vert
\psi\right\rangle \geq k+\delta,
\]
we need%
\[
\Omega\left(  \frac{\delta^{2}\sqrt{2^{n}}}{\ell^{2}k\log k}-\ell\right)
\]
queries to $U_{\psi}$.
\end{theorem}

Intriguingly, Theorem \ref{ctnc0}\ can be seen as a common generalization of
the No-Cloning Theorem and the BBBV lower bound for quantum search
\cite{bbbv}.\footnote{Here by \textquotedblleft quantum
search,\textquotedblright\ we mean search for an unknown pure state
$\left\vert \psi\right\rangle $, which need not be a computational basis
state. \ As far as we know, this generalization of the usual Grover
problem\ was first studied by Farhi and Gutmann \cite{farhigutmann}.} \ It
reduces to the No-Cloning Theorem if we ignore the oracle $U$, and it reduces
to the BBBV lower bound if we ignore the initial state $\left\vert
\psi\right\rangle $.

The proof of Theorem \ref{ctnc0}\ proceeds in two steps. \ First we
lower-bound the query complexity of cloning $\left\vert \psi\right\rangle $
almost \textit{perfectly}, by using a generalization of Ambainis's quantum
adversary method \cite{ambainis} that we design specifically for the purpose.
\ Next we argue that, if we could even clone $\left\vert \psi\right\rangle
$\ with non-negligible fidelity, then with polynomially more queries we could
also clone $\left\vert \psi\right\rangle $ almost perfectly, by using a recent
fixed-point quantum search algorithm of Grover \cite{grover:fps}.

We regret that, due to space limitations, we are not able to include a proof
of Theorem \ref{ctnc0}\ in this extended abstract.

With Theorem \ref{ctnc0} in hand, it is not hard to show the existence of a
quantum oracle $U$\ relative to which a publicly-verifiable quantum money
scheme exists. \ We simply choose $n$-qubit quantum banknotes uniformly at
random under the Haar measure, and then \textquotedblleft
offload\textquotedblright\ all the work of preparing and recognizing the
banknotes onto the oracle. \ Theorem \ref{ctnc0} then implies that, even given
$k=\operatorname*{poly}\left(  n\right)  $ valid banknotes, a would-be
counterfeiter needs exponentially many queries to $U$ to prepare a $\left(
k+1\right)  ^{\text{st}}$\ banknote. \ Crucially, our oracle construction is
\textquotedblleft fair,\textquotedblright\ in the sense that the bank, the
customers, and the counterfeiters all have access to the same oracle $U$, and
none of them have any special knowledge about $U$ not shared by the others.
\ This is why we believe our result merits the informal interpretation we have
given it: namely, that any impossibility proof for quantum money would have to
be non-relativizing.

Showing the existence of a quantum oracle $U$ relative to which quantum
copy-protection works is a harder problem. \ As in the money case, we choose
$n$-qubit \textquotedblleft quantum programs\textquotedblright\ $\left\vert
\psi_{f}\right\rangle $\ uniformly at random under the Haar measure, and then
define a quantum oracle $U$ that is able both to prepare $\left\vert \psi
_{f}\right\rangle $\ given a description of $f$, and to evaluate $f\left(
x\right)  $\ given $\left\vert \psi_{f}\right\rangle $ and $x$. \ However, a
new difficulty is that some families of Boolean functions $\mathcal{F}$
\textit{cannot} be copy-protected: namely, those for which any $f\in
\mathcal{F}$\ can be efficiently learned using black-box access. \ Thus, our
proof somehow needs to explain why learnability is the only obstruction. \ Our
solution will be to construct a polynomial-time \textit{simulator}, which
takes an algorithm (in the oracle world)\ for pirating a quantum program
$\left\vert \psi_{f}\right\rangle $, and converts it into an algorithm (with
no oracle) that learns $f$ using only black-box access to $f$.

Among other things, the simulator needs the ability to \textquotedblleft mock
up\textquotedblright\ its own quantum state $\left\vert \varphi\right\rangle $
that can stand in\ for $\left\vert \psi_{f}\right\rangle $\ in a simulation of
the pirating algorithm, which in turn means that $\left\vert \varphi
\right\rangle ^{\otimes t}$\ should be indistinguishable from $t$ copies of a
Haar-random state for some fixed $t=\operatorname*{poly}\left(  n\right)  $.
\ As it turns out, precisely this problem---the construction of explicit
quantum states that behave like Haar-random states---has recently become a
major topic in quantum computing. \ So for example, Ambainis and Emerson
\cite{ae} gave an explicit construction of \textit{approximate quantum }%
$t$\textit{-designs}\ for arbitrary $t$: that is, finite ensembles of pure
states $\left(  p_{x},\left\vert \varphi_{x}\right\rangle \right)  $ such that%
\begin{align*}
\left(  1-\varepsilon\right)  \int_{\psi}\left(  \left\vert \psi\right\rangle
\left\langle \psi\right\vert \right)  ^{\otimes t}d\psi &  \leq\sum_{x}%
p_{x}\left(  \left\vert \varphi_{x}\right\rangle \left\langle \varphi
_{x}\right\vert \right)  ^{\otimes t}\\
&  \leq\left(  1+\varepsilon\right)  \int_{\psi}\left(  \left\vert
\psi\right\rangle \left\langle \psi\right\vert \right)  ^{\otimes t}d\psi
\end{align*}
where the integrals are with respect to the Haar measure. \ Our requirement is
slightly different: basically, we need that no algorithm that receives $t$
copies of $\left\vert \varphi\right\rangle $, \textit{and} makes $T$ queries
to an oracle that recognizes $\left\vert \varphi\right\rangle $, can decide
whether $\left\vert \varphi\right\rangle $\ was drawn from the explicit
distribution or the Haar measure. \ Both for that reason, and because our
construction was independent of \cite{ae}, in the full version of the paper we
give a self-contained proof of the following result:

\begin{theorem}
\label{tdesthm}Let $d$ be a positive integer. \ Then there exists a collection
of $n$-qubit pure states $\left\{  \left\vert \varphi_{x}\right\rangle
\right\}  _{x\in\left\{  0,1\right\}  ^{n\left(  d+1\right)  }}$\ such that:

\begin{enumerate}
\item[(i)] Given $x$ as input, the state $\left\vert \varphi_{x}\right\rangle
$\ can be prepared in time polynomial in $n$ and $d$.

\item[(ii)] Let $E$\ be any quantum algorithm that receives a state
$\left\vert \varphi\right\rangle ^{\otimes t}$\ as input, and also makes
$T$\ queries to a quantum oracle $U_{\varphi}$\ such that $U_{\varphi
}\left\vert \varphi\right\rangle =-\left\vert \varphi\right\rangle $ and
$U_{\varphi}\left\vert \phi\right\rangle =\left\vert \phi\right\rangle $\ for
all $\left\vert \phi\right\rangle $\ orthogonal to $\left\vert \varphi
\right\rangle $. \ Let $E\left(  \left\vert \varphi\right\rangle \right)
$\ represent the probability that $E$\ accepts. \ Then provided $t+2T\leq
\min\left\{  d/2,\sqrt{2^{n}/2}\right\}  $, we have%
\[
\left\vert
\begin{array}
[c]{c}%
\operatorname*{EX}_{x\in\left\{  0,1\right\}  ^{n\left(  d+1\right)  }}\left[
E\left(  \left\vert \varphi_{x}\right\rangle \right)  \right] \\
-\operatorname*{EX}_{\left\vert \psi\right\rangle \in\mu}\left[  E\left(
\left\vert \psi\right\rangle \right)  \right]
\end{array}
\right\vert \leq\frac{4\left(  t+2T\right)  ^{2}}{2^{n}}%
\]
where $\mu$\ is the Haar measure.
\end{enumerate}
\end{theorem}

For those who are curious, the explicit states in question are%
\[
\left\vert \varphi_{x}\right\rangle :=\frac{1}{\sqrt{2^{n}}}\sum
_{x\in\mathbb{GF}\left(  2^{n}\right)  }e^{2\pi ip\left(  x\right)  /2^{n}%
}\left\vert x\right\rangle ,
\]
where $p:\mathbb{GF}\left(  2^{n}\right)  \rightarrow\mathbb{GF}\left(
2^{n}\right)  $ is a univariate polynomial of degree at most $d$\ that is
encoded by the string $x\in\left\{  0,1\right\}  ^{n\left(  d+1\right)  }$,
and elements of $\mathbb{GF}\left(  2^{n}\right)  $\ are freely reinterpreted
as $n$-bit integers where relevant.

\subsection{Related Work\label{PREV}}

Recall that quantum money was first studied by Wiesner \cite{wiesner}. \ In
Wiesner's scheme, a central bank distributes \textquotedblleft quantum
banknotes,\textquotedblright\ each consisting of a unique serial number (which
is written down classically), together with $n$ polarized photons in the
states $\left\vert 0\right\rangle $, $\left\vert 1\right\rangle $, $\left\vert
+\right\rangle =\frac{\left\vert 0\right\rangle +\left\vert 1\right\rangle
}{\sqrt{2}}$, or $\left\vert -\right\rangle =\frac{\left\vert 0\right\rangle
+\left\vert 1\right\rangle }{\sqrt{2}}$. \ The bank also stores, in a secure
location, a database of all the serial numbers together with classical
descriptions of the associated quantum states. \ Whenever a banknote is
returned to the bank, the note can be measured (using the secure database) to
verify its authenticity. \ On the other hand, using the uncertainty principle,
it is possible to show that, starting from $k$\ banknotes, any attempt to
forge $k+1$\ banknotes that all pass the authentication test can succeed with
probability at most $\left(  3/4\right)  ^{n}$.

Let us point out two striking advantages of Wiesner's scheme. \ Firstly, the
scheme requires only single coherent qubits and one-qubit measurements; there
is no need for any entanglement. \ For this reason, the scheme might be
practical long before universal quantum computing. \ Secondly, the security of
the scheme is \textit{information-theoretic}---guaranteed by the laws of
quantum physics---rather than computational.

An obvious drawback of Wiesner's scheme is its need for a giant secret
database maintained by the bank. \ But in 1982, Bennett, Brassard, Breidbart,
and Wiesner\ \cite{bbbw} (henceforth BBBW) showed how to avoid the giant
database, at the cost of making the security of the quantum money
computational rather than information-theoretic. \ In modern terms, their
proposal was this. \ The bank fixes, once and for all, a secret random seed
$s$. \ It then distributes banknotes, each of the form $\left\vert
y\right\rangle \left\vert \psi_{g_{s}\left(  y\right)  }\right\rangle
$,\ where $y\in\left\{  0,1\right\}  ^{n}$\ is a unique serial number for the
banknote, $g_{s}:\left\{  0,1\right\}  ^{n}\rightarrow\left\{  0,1\right\}
^{n}$\ is a pseudorandom function, and $\left\vert \psi_{g_{s}\left(
y\right)  }\right\rangle $\ is the state obtained by starting from
$g_{s}\left(  y\right)  $, grouping the $n$ bits into $n/2$\ blocks of two,
and mapping each $00$\ to $\left\vert 0\right\rangle $, $01$\ to $\left\vert
1\right\rangle $,\ $10$\ to $\left\vert +\right\rangle $, and $11$\ to
$\left\vert -\right\rangle $.

Using its knowledge of $s$, the bank can verify the authenticity of any note
$\left\vert y\right\rangle \left\vert \psi_{g_{s}\left(  y\right)
}\right\rangle $, by computing $g_{s}\left(  y\right)  $\ and then measuring
each qubit of $\left\vert \psi_{g_{s}\left(  y\right)  }\right\rangle $\ in
the appropriate basis. \ But suppose $g_{s}$\ were a truly random function.
\ Then by the same argument as for Wiesner's original scheme, given any $k$
banknotes, no quantum operation could forge a $\left(  k+1\right)
^{\text{st}}$\ note with probability more than $\left(  3/4\right)  ^{n/2}%
$\ of passing the authentication test. \ This means that, if there
\textit{were} a quantum operation to forge high-quality banknotes, then that
operation could be used to distinguish $g_{s}$\ from a truly random function.
\ And therefore, assuming $g_{s}$ is secure against polynomial-time quantum
adversaries, there can be no polynomial-time quantum algorithm to forge
banknotes that pass the authentication test with non-negligible probability.

However, the BBBW scheme still has a serious drawback: namely that $s$, which
is needed for the authentication procedure, must remain a closely-guarded
secret. \ And thus it would presumably be unwise, for example, to install the
authentication devices in convenience-store cash registers. \ What we really
want is a scheme where the procedure for \textit{authenticating} the money is
completely public, and only the procedure for minting the money is secret.

Bennett et al.\ \cite{bbbw} presented a candidate for such a
publicly-verifiable quantum money scheme,\footnote{Bennett et al.\ described
their public-key scheme in terms of \textquotedblleft subway
tokens\textquotedblright\ rather than money---since if we want to authenticate
the tokens using single-qubit measurements only, then the authentication test
necessarily destroys the tokens and prevents their reuse. \ On the other hand,
supposing we could perform an entangled measurement on all $n$ qubits in a
token, it would be possible to authenticate the token while preserving its
quantum coherence. \ For this reason, the token could be used as money.} which
was based on the hardness of factoring Blum integers.\footnote{More generally,
their scheme could be based on any \textit{trapdoor collision-resistant hash
function}: that is, a CRHF such that one can efficiently sample collision
pairs using some hidden trapdoor information.} \ Unfortunately, their scheme
was insecure for two reasons. \ First, we now know that factoring is in
quantum polynomial time! \ But even were we to base the scheme on some other
cryptographic primitive, Bennett et al.\ pointed out that it could be broken
by an adversary who is able to make entangled measurements on all of the
qubits in a banknote. \ The question of whether secure quantum money with
public authentication is possible has remained open for $30$ years.

Concurrently with our work, there has been a recent renewal of interest in the
quantum money problem. \ In his PhD thesis, Stebila \cite{stebila} provides a
lucid overview of quantum money, and explains why our Complexity-Theoretic
No-Cloning Theorem implies the existence of a quantum oracle relative to which
publicly-verifiable quantum money is possible.\footnote{Indeed, our original
interest was in copy-protection; it was Stebila, along with M. Mosca, who
pointed out to us the application to unforgeable money.}

As far as we know, the idea of using quantum mechanics to copy-protect
software\ is original to this work.

\subsection{Organization\label{ORG}}

The rest of this extended abstract is organized as follows. \ Section
\ref{PRELIM} formally defines quantum money and copy-protection schemes and
investigates their basic properties, and also recalls some preliminaries from
cryptography and quantum information. \ Section \ref{MONEYINT} considers
quantum money schemes: our explicit candidate proposal based on random
stabilizer states in Section \ref{MONEYSCHEME}, and our oracle result in
Section \ref{MONEYORACLE}. \ Section \ref{CPINT} then discusses quantum
copy-protection: the candidate schemes for copy-protecting point functions in
Section \ref{POINTFN}, and the oracle result in Section \ref{CPORACLE}. \ We
conclude in Section \ref{OPEN}\ with a list of open problems. \ We regret
that, because of space limitations, much of the paper's technical content
(including the proof of the Complexity-Theoretic No-Cloning Theorem and the
explicit construction of quantum $t$-designs) has had to be relegated to the
full version.

\section{Preliminaries\label{PRELIM}}

For simplicity, in this paper we restrict ourselves to nonuniform (circuit)
computation. \ Given two mixed states $\rho$\ and $\sigma$, the \textit{trace
distance} $\left\Vert \rho-\sigma\right\Vert _{\operatorname*{tr}}$\ equals
the maximum, over all measurements $M$, of the variation distance $\left\Vert
M\left(  \rho\right)  -M\left(  \sigma\right)  \right\Vert $ between the
probability distributions $M\left(  \rho\right)  ,M\left(  \sigma\right)
$\ over measurement outcomes obtained by applying $M$ to $\rho$\ and $\sigma
$\ respectively. \ We will use the following lemma of Aaronson \cite{aar:adv}:

\begin{lemma}
[\textquotedblleft Almost As Good As New Lemma\textquotedblright%
]\label{goodasnew}Suppose a measurement on a mixed state $\rho$\ yields a
particular outcome with probability $1-\varepsilon$. \ Then after the
measurement, one can recover a state $\widetilde{\rho}$\ such that $\left\Vert
\widetilde{\rho}-\rho\right\Vert _{\operatorname*{tr}}\leq\sqrt{\varepsilon}$.
\end{lemma}

See Nielsen and Chuang \cite{nc}\ for other quantum information concepts used
in this paper.

In what follows, we will sometimes use the assumption that there exists a
pseudorandom function family secure against quantum adversaries. \ The
following theorem helps to justify that assumption.

\begin{theorem}
\label{quantumprf}Suppose there exists a one-way function $A:\left\{
0,1\right\}  ^{n}\rightarrow\left\{  0,1\right\}  ^{n}$\ that is secure
against $2^{n^{\Omega\left(  1\right)  }}$-time quantum adversaries. \ Then
there also exists a family $f_{s}:\left\{  0,1\right\}  ^{n}\rightarrow
\left\{  0,1\right\}  ^{n}$\ of pseudorandom functions, parameterized by a
seed $s\in\left\{  0,1\right\}  ^{\operatorname*{poly}\left(  n\right)  }$,
that is secure against $2^{n}$-time quantum adversaries. \ (Here $A$\ and
$f_{s}$\ are both computable in classical polynomial time.)
\end{theorem}

\begin{proof}
[Proof Sketch]H\aa stad et al.\ \cite{hill}\ showed that if $2^{n^{\Omega
\left(  1\right)  }}$-secure one-way functions exist, then so do
$2^{n^{\Omega\left(  1\right)  }}$-secure\ pseudorandom
generators.\ \ Razborov and Rudich \cite{rr}\ showed that if $2^{n^{\Omega
\left(  1\right)  }}$-secure pseudorandom generators exist, then so do $2^{n}%
$-secure\ pseudorandom function families (with polynomial seed length).
\ Since both of these reductions are \textquotedblleft
black-box,\textquotedblright\ they go through essentially without change if
the adversary is quantum.
\end{proof}

Interestingly, the reduction of Goldreich, Goldwasser, and Micali \cite{ggm},
from $f\left(  n\right)  $-secure pseudorandom\ generators to $f\left(
n\right)  ^{\Omega\left(  1\right)  }/\operatorname*{poly}\left(  n\right)
$-secure pseudorandom functions with seed length $n$, does \textit{not} go
through if the adversary is quantum.\footnote{This is because the GGM
reduction makes essential use of the fact that a polynomial-time adversary can
examine only $\operatorname*{poly}\left(  n\right)  $ outputs of the function
$f_{s}$---something that is manifestly false if the adversary is quantum.}
\ We leave as an open problem whether a \textquotedblleft
strong,\textquotedblright\ GGM-style reduction from PRGs to PRFs can be proved
in the quantum setting.

\subsection{Quantum Money\label{MONEYPRELIM}}

Intuitively, a public-key quantum money scheme is a scheme in which

\begin{enumerate}
\item[(1)] quantum banknotes can be efficiently produced by a central bank,

\item[(2)] there exists a polynomial-time quantum algorithm for authenticating
the banknotes, which is completely public, and

\item[(3)] given as input $k$ valid banknotes, a polynomial-time counterfeiter
cannot produce $k+1$\ valid banknotes that have non-negligible probability of
passing the authentication test.
\end{enumerate}

We now give a formal definition.

\begin{definition}
\label{moneydef}A quantum money scheme with key size $n$ consists of the following:

\begin{itemize}
\item A quantum circuit $B$ of size $O\left(  \operatorname*{poly}\left(
n\right)  \right)  $ (the \textquotedblleft bank\textquotedblright), which
takes a string $s\in\left\{  0,1\right\}  ^{n}$\ (the \textquotedblleft secret
key\textquotedblright) as input, and produces a classical string $e_{s}$\ (the
\textquotedblleft public key\textquotedblright) and mixed state $\rho_{s}%
$\ (the \textquotedblleft banknote\textquotedblright) as output.\footnote{One
can of course generalize the definition to let $e_{s}$\ be randomized or even
a quantum state, and possibly correlated with $\rho_{s}$\ as well. \ However,
we will not need the additional freedom in this paper.}

\item A quantum circuit $A$ of size $O\left(  \operatorname*{poly}\left(
n\right)  \right)  $ (the \textquotedblleft authenticator\textquotedblright),
which takes a string $e$\ and state $\rho$\ as input and either accepts or rejects.
\end{itemize}

We say $\left(  B,A\right)  $\ has completeness error $\varepsilon$ if
$A\left(  e_{s},\rho_{s}\right)  $ accepts with probability at least
$1-\varepsilon$\ for all $s$. \ We say $\left(  B,A\right)  $\ has soundness
error $\delta$\ if for all quantum circuits $C$ of size $O\left(
\operatorname*{poly}\left(  n\right)  \right)  $ (the \textquotedblleft
counterfeiter\textquotedblright) and all $k,r=O\left(  \operatorname*{poly}%
\left(  n\right)  \right)  $, the following holds. \ Assume $C$ takes
$\rho_{s}^{\otimes k}$ as input, and outputs a state $\sigma_{s}$\ on
$k+r$\ registers. \ For $i\in\left[  k+r\right]  $, let $\sigma_{s}^{i}%
$\ denote the contents of the $i^{\text{th}}$\ register, and let $p_{i}$\ be
the probability that $A\left(  e_{s},\sigma_{s}^{i}\right)  $\ accepts,
averaged over all $s\in\left\{  0,1\right\}  ^{n}$. \ Then $\sum_{i=1}%
^{k+r}p_{i}\leq k+\delta$.

We call $\left(  B,A\right)  $\ public-key if $C$ also receives $e_{s}$\ as
input, and private-key otherwise. \ If $\left(  B,A\right)  $\ is private-key,
we call it query-secure if $C$\ has access to an oracle that takes a state
$\sigma$\ as input and simulates $A\left(  e_{s},\sigma\right)  $ (that is,
accepts with the same probability and returns the same post-measurement state
$\widetilde{\sigma}$).\footnote{Note that any public-key scheme is also
query-secure, since we can hardwire a description of $A$ into $C$.}
\end{definition}

We make a few remarks on Definition \ref{moneydef}. \ First, it is obvious
that no money scheme exists where the states $\rho_{s}$\ are
classical.\footnote{Furthermore, no query-secure scheme can exist where the
states $\rho_{s}$\ have $O\left(  \log n\right)  $\ qubits. \ For in that
case, a counterfeiter could reconstruct $\rho_{s}$ in polynomial time, by
first generating a tomographically complete set of states, and then sending
several copies of each state to $A$ to estimate the probability that each one
is accepted.} \ Second, if a money scheme has completeness error $\varepsilon
$, it follows from Lemma \ref{goodasnew} that the authentication procedure can
return a banknote $\widetilde{\rho}_{s}$\ such\ that $\left\Vert
\widetilde{\rho}_{s}-\rho_{s}\right\Vert \leq\sqrt{\varepsilon}$. \ This means
that the same banknote can be verified $\Omega\left(  1/\sqrt{\varepsilon
}\right)  $\ times before it needs to be replaced. \ In this paper, we will
generally be interested in schemes with perfect completeness.

Third, we will generally want the soundness error $\delta$\ to be negligible
(that is, $o\left(  1/p\left(  n\right)  \right)  $\ for all polynomials $p$).
\ If $\delta$\ is negligible, then it is easy to see that, starting from
$\rho_{s}^{\otimes k}$, no polynomial-time counterfeiter $C$ can ever increase
its \textquotedblleft wealth\textquotedblright\ (defined as the expected
number of states in $C$'s possession that $A$ accepts) by more than a
negligible amount in expectation. \ Note that this is true even if the states
output by $C$ are entangled; our definition automatically accounts for this possibility.

We now discuss some examples. \ The BBBW scheme \cite{bbbw}, discussed in
Section \ref{PREV}, is a private-key quantum money scheme. \ We therefore have
the following:

\begin{theorem}
[implicit in \cite{bbbw}]\label{bbbwthm}If there exists a pseudorandom
function family secure against quantum adversaries, then there exists a
private-key quantum money scheme with perfect completeness and exponentially
small soundness error.
\end{theorem}

However, the BBBW scheme is \textit{not} query-secure. \ The reason is simple:
given a banknote of the form $\left\vert y\right\rangle \left\vert \psi
_{g_{s}\left(  y\right)  }\right\rangle $, a counterfeiter can learn a
classical description of $\left\vert \psi_{g_{s}\left(  y\right)
}\right\rangle $, by rotating each qubit $i$ in turn while leaving the other
$n/2-1$\ qubits fixed, and repeatedly feeding the result to the
authenticator\ $A$\ until it has ascertained the correct state of the $i^{th}%
$\ qubit. \ This works because $A$ always measures the qubits in the correct
bases, and therefore does not damage the qubits that are not being rotated.
\ Of course, once the counterfeiter has learned a classical description of
$\left\vert \psi_{g_{s}\left(  y\right)  }\right\rangle $, it can then produce
as many copies of $\left\vert y\right\rangle \left\vert \psi_{g_{s}\left(
y\right)  }\right\rangle $\ as it likes.

In the full version of this paper, we will give a private-key quantum money
scheme that \textit{is} query-secure, assuming the existence of pseudorandom
functions secure against quantum adversaries. \ We will also prove the
following result, which is not entirely trivial:

\begin{theorem}
\label{supercounterfeit}Any quantum money scheme satisfying Definition
\ref{moneydef} (even a private-key one) must rely on \textit{some}
computational assumption.
\end{theorem}

As mentioned in Section \ref{PREV}, Wiesner's original scheme \cite{wiesner}
avoided the need for any computational assumption, but only by having the bank
maintain a giant lookup table, containing a classical description of every
banknote that has ever been issued. \ If we want to fit Wiesner's scheme into
Definition \ref{moneydef}, one way to do it is to assume that all parties have
access to a (classical) random oracle $\mathcal{O}$. \ For then the bank can
use a secret part of the oracle string to generate the banknotes $\left\vert
y\right\rangle \left\vert \psi_{y}\right\rangle $;\ and to any counterfeiter
who does not know which part of the oracle the bank is using, the states
$\left\vert \psi_{y}\right\rangle $\ will appear to be drawn uniformly from
$\left\{  \left\vert 0\right\rangle ,\left\vert 1\right\rangle ,\left\vert
+\right\rangle ,\left\vert -\right\rangle \right\}  ^{n}$. \ This observation
gives us the following:

\begin{theorem}
[implicit in \cite{wiesner}]\label{wiesnerthm}Relative to a random oracle
$\mathcal{O}$, there exists a private-key quantum money scheme with perfect
completeness and exponentially small soundness error.
\end{theorem}

On the other hand, Wiesner's scheme is not query-secure, for the same reason
the BBBW scheme is not. \ In the full version of this paper, we give a
private-key quantum money scheme that \textit{is} query-secure, relative to a
random oracle $\mathcal{O}$.

In Section \ref{MONEYSCHEME}, we will present a candidate for a public-key
quantum money scheme based on random stabilizer states, while in Section
\ref{MONEYORACLE}, we will prove that public-key quantum money schemes exist
relative to a quantum oracle.

The situation is summarized in Table \ref{MONEYTAB}.

\begin{table*}[ptb]%
\begin{tabular}
[c]{l|lllll}%
\textbf{Money Scheme} & \textbf{Type} & \textbf{Oracle} & \textbf{Security} &
\textbf{States Used} & \textbf{Reference}\\\hline
Wiesner & Private-key & Random & Unconditional & Single qubits &
\cite{wiesner}\\
BBBW & Private-key & None & Assuming PRFs & Single qubits & \cite{bbbw}\\
Modified Wiesner & Query-secure & Random & Unconditional & Haar-random & Full
version\\
Modified BBBW & Query-secure & None & Assuming PRFs & Haar-random & Full
version\\
Quantum Oracle & Public-key & Quantum & Unconditional & Haar-random & This
paper\\
Random Stabilizers & Public-key & None & Conjectured & Stabilizer & This
paper
\end{tabular}
\label{MONEYTAB}\caption{Known quantum money schemes and their properties}%
\end{table*}

\subsection{Quantum Copy-Protection\label{CPPRELIM}}

What if we want to distribute unclonable quantum states that are
\textit{useful} for something besides just getting authenticated? \ This
brings us to the question of \textit{quantum software copy-protection}.
\ Informally, given a secret Boolean function $f:\left\{  0,1\right\}
^{n}\rightarrow\left\{  0,1\right\}  $ drawn from a known family $\mathcal{F}%
$, what we want is a quantum state $\rho_{f}$\ that

\begin{enumerate}
\item[(1)] can be efficiently prepared given a classical description of $f$,

\item[(2)] can be used to compute $f\left(  x\right)  $\ efficiently for any
input $x\in\left\{  0,1\right\}  ^{n}$, and

\item[(3)] \textit{cannot} be efficiently used to prepare more states from
which $f$\ can be computed in quantum polynomial time.
\end{enumerate}

It is clear that, if the function family $\mathcal{F}$\ is efficiently
\textit{learnable}---in the sense that we can output a circuit for an unknown
$f\in\mathcal{F}$ in quantum polynomial time, using only oracle access to
$f$---then there is no hope of copy-protecting $f$. \ For in that case, being
able to run a program for $f$ is tantamount to being able to copy the program.
\ Indeed, even if we cannot learn a useful \textit{classical} description of
$f$ by measuring $\rho_{f}$, it might still be possible to prepare additional
quantum programs for $f$ directly, by some quantum operation on $\rho_{f}$.

The quantum copy-protection problem might remind readers of the classical
\textit{code obfuscation} problem, and indeed there are similarities.
\ Roughly speaking, we say a program $P$ for a function $f\in\mathcal{F}$\ is
obfuscated if knowing $P$'s source code is \textquotedblleft no more
useful\textquotedblright\ than being able to run $P$, in the sense that any
property of $f$ that is efficiently computable given $P$'s source code, is
also efficiently computable given oracle access to $f$. \ Barak et
al.\ \cite{baraketal} famously showed that there exist function families
$\mathcal{F}$\ that are impossible to obfuscate. \ On the other hand, Wee
\cite{wee} and others\ have shown that, under strong cryptographic
assumptions, it \textit{is} possible to obfuscate point functions and several
related families of functions. \ In Section \ref{POINTFN}, we will give
proposals for quantumly copy-protecting point functions that are somewhat
reminiscent of known methods for obfuscating point functions.

However, let us point out two differences between copy-protection and
obfuscation. \ Firstly, it is trivial to show that copy-protection is
\textit{always} impossible in the classical world, for \textit{any} function
family: one does not need anything like the elegant argument of Barak et
al.\ \cite{baraketal}. \ Secondly, as discussed before, any function family
$\mathcal{F}$\ that is learnable from input/output behavior cannot be
copy-protected---but for exactly the same reason, $\mathcal{F}$\ \textit{can}
be obfuscated! \ For if we can output a program for $f\in\mathcal{F}$\ using
only oracle access to $f$, then clearly the source code of that program is no
more useful than the oracle access. \ Thus, while unbreakable copy-protection
has connections with obfuscation, fundamentally it is a new cryptographic
task, one whose very possibility depends on quantum mechanics.

We now define quantum copy-protection schemes.

\begin{definition}
\label{copydef}Consider a family $\mathcal{F}$\ of Boolean functions
$f:\left\{  0,1\right\}  ^{n}\rightarrow\left\{  0,1\right\}  $, where each
$f\in\mathcal{F}$\ is associated with a unique \textquotedblleft
description\textquotedblright\ $d_{f}\in\left\{  0,1\right\}  ^{m}$. \ (Thus
$\left\vert \mathcal{F}\right\vert \leq2^{m}$.) \ A quantum copy-protection
scheme for $\mathcal{F}$ consists of the following:

\begin{itemize}
\item A quantum circuit $V$\ of size $O\left(  \operatorname*{poly}\left(
n,m\right)  \right)  $ (the \textquotedblleft vendor\textquotedblright), which
takes $d_{f}$\ as input and produces a mixed state $\rho_{f}$\ as output.

\item A quantum circuit $C$\ of size $O\left(  \operatorname*{poly}\left(
n,m\right)  \right)  $ (the \textquotedblleft customer\textquotedblright),
which takes $\left(  \rho_{f},x\right)  $\ as input and attempts to output
$f\left(  x\right)  $.
\end{itemize}

We say $\left(  V,C\right)  $\ has correctness parameter $\varepsilon$\ if $C$
outputs $f\left(  x\right)  $\ with probability at least $1-\varepsilon$ given
$\left(  \rho_{f},x\right)  $\ as input, for all $f\in\mathcal{F}$\ and
$x\in\left\{  0,1\right\}  ^{n}$.

We say $\left(  V,C\right)  $\ has security $\delta$\ against a probability
distribution $\mathcal{D}$\ over $\mathcal{F}\times\left\{  0,1\right\}  ^{n}%
$,\ if for all quantum circuits $P$ and $L$\ of size $O\left(
\operatorname*{poly}\left(  n,m\right)  \right)  $ (the \textquotedblleft
pirate\textquotedblright\ and \textquotedblleft freeloader\textquotedblright%
\ respectively) and all $k,r=O\left(  \operatorname*{poly}\left(  n,m\right)
\right)  $, the following holds. \ Assume $P$ takes\ $\rho_{f}^{\otimes k}%
$\ as input, and outputs a state $\sigma_{f}$\ on $k+r$\ registers. \ For
$i\in\left[  k+r\right]  $, let $\sigma_{f}^{i}$\ denote the contents of the
$i^{\text{th}}$\ register. \ Also, suppose $L$ takes $\left(  \rho
_{f},x\right)  $\ as input and attempts to output $f\left(  x\right)  $.
\ Then if we run $L$\ on $(\sigma_{f}^{i},x)$ for all $i\in\left[  k+r\right]
$, the expected number of invocations that output $f\left(  x\right)  $,
averaged over $\left(  f,x\right)  $\ drawn from $\mathcal{D}$, is at most
$k+\left(  1-\delta\right)  r$.\footnote{One might also want to require a
concentration inequality---e.g. that for all inputs $x$,\ the probability that
at least $k+2r/3$\ of the pirated programs output $f\left(  x\right)
$\ correctly decreases exponentially with $r$. \ This is a topic we leave to
future work.}
\end{definition}

A few remarks on Definition \ref{copydef}. \ First, the security criterion
might seem a bit strange. \ The basic motivation is that we need to ignore
\textquotedblleft trivial\textquotedblright\ pirating strategies, such as
mapping the state\ $\rho_{f}^{\otimes2}$\ to%
\[
\frac{1}{3}\left(  \rho_{f}\otimes\rho_{f}\otimes I+\rho_{f}\otimes
I\otimes\rho_{f}+I\otimes\rho_{f}\otimes\rho_{f}\right)  ,
\]
which has large fidelity with $\rho_{f}$\ on each of the three registers. \ On
the other hand, we also do not want to require all $k+r$\ pirated programs to
output the right answer \textit{simultaneously} (with high probability and on
some input $x$), since that criterion is too stringent even for legitimate
programs with constant error. \ Looking at the expected number of correct
answers is convenient, since by linearity of expectation, we can then ignore
entanglement and classical correlations among the registers. \ Note that it is
always possible for $\left(  1-\varepsilon\right)  k+r/2$ of the
$k+r$\ pirated programs to get the right answers on average---using a pirating
strategy that outputs the legitimate programs $\rho_{f}^{\otimes k}$,
alongside $r$ programs that guess randomly on every input $x$. \ But ideally
it should not be possible to do too much better than that.

Second, a natural question is whether the state $\rho_{f}$ can be used more
than once, or whether the irreversibility of measurement makes such a state
\textquotedblleft disposable.\textquotedblright\ \ In our setting, disposable
states might actually be \textit{preferred}---since any disposable state is
copy-protected by definition! \ (If we could copy $\rho_{f}$\ with high
fidelity, then we could run each copy on a different input $x$, contrary to
assumption.) \ However, it is not hard to see that, provided the customer buys
$k=\Omega\left(  n\right)  $ copies of $\rho_{f}$\ from the quantum software
store, she can evaluate $f$ on as many inputs as she likes---indeed, all
$2^{n}$\ of them, if she has exponential time. \ For by standard
amplification, $\rho_{f}^{\otimes k}$\ can be used to evaluate $f$ with error
probability $2^{-\Omega\left(  k\right)  }$. \ So by Lemma \ref{goodasnew}, it
is possible to reuse $\rho_{f}^{\otimes k}$\ an exponential number of times,
by uncomputing garbage after each measurement.

In this paper, we will typically assume that $\rho_{f}$ \textquotedblleft
comes from the store\textquotedblright\ already amplified, and that both
customers and would-be software pirates can therefore reuse $\rho_{f}$ as many
times as needed. \ This raises an interesting point: given an amplified state
$\rho_{f}=\sigma^{\otimes k}$, a customer willing to tolerate slightly higher
error could always split $\rho_{f}$ into $\sigma^{\otimes k/2}\otimes
\sigma^{\otimes k/2}$,\ and give one of the copies of $\sigma^{\otimes k/2}$
to a friend (rather like donating a kidney). \ We leave as an open question
whether it is possible to amplify success probability in a way that does not
allow this sort of sharing.

Third, call the function family $\mathcal{F}$ and distribution $\mathcal{D}%
$\ \textit{quantumly learnable with error }$\delta$ if there exist
polynomial-size quantum circuits $Q$ and $C$ such that%
\[
\Pr_{\left(  f,x\right)  \in\mathcal{D}}\left[  C\left(  Q^{f},x\right)
\text{ outputs }f\left(  x\right)  \right]  \geq1-\delta,
\]
where $Q^{f}$ denotes the mixed state output by $Q$ given oracle access to
$f$. \ (Note that $Q$ does not receive $x$.) \ The following simple
proposition delimits the function families\ that one can hope to copy-protect.

\begin{proposition}
\label{learncp}No $\left(  \mathcal{F},\mathcal{D}\right)  $\ pair that is
quantumly learnable with error $\delta$\ can be quantumly copy-protected with
security $\delta+2^{-n}$.
\end{proposition}

\begin{proof}
Using an amplified state of the form $\rho_{f}^{\otimes\operatorname*{poly}%
\left(  n\right)  }$, a pirate can \textit{simulate} quantum oracle access to
$f$ with exponentially small error. \ The pirate can thereby use the learning
algorithm $Q^{f}$ to output as many states $\sigma_{f}$\ as he wants with the
property that $\Pr_{\left(  f,x\right)  \in\mathcal{D}}\left[  C\left(
\sigma_{f},x\right)  \text{ outputs }f\left(  x\right)  \right]  \geq
1-\delta-2^{-n}$. \ Note that by Lemma \ref{goodasnew}, each \textquotedblleft
query\textquotedblright\ to $f$ damages $\rho_{f}^{\otimes\operatorname*{poly}%
\left(  n\right)  }$\ by only an exponentially small amount.
\end{proof}

Notice that if $\left\vert \mathcal{F}\right\vert \leq\operatorname*{poly}%
\left(  n\right)  $, then $\mathcal{F}$\ is quantumly learnable (and indeed
classically learnable), since the learning algorithm $Q$\ simply needs to
hardwire inputs $x_{1},\ldots,x_{\left\vert \mathcal{F}\right\vert -1}$\ such
that every distinct $f,f^{\prime}\in\mathcal{F}$\ differ on some $x_{i}$.
\ Thus, one corollary of Proposition \ref{learncp}\ is that we can only hope
to copy-protect superpolynomially large function families.

Let us end with a simple but important fact, which shows that, as in the
quantum money case, we can only hope for security under computational assumptions.

\begin{proposition}
\label{superpirate}A software pirate with unlimited computational power can
break any quantum copy-protection scheme.
\end{proposition}

\begin{proof}
Let $f$ and $g$ be two functions in $\mathcal{F}$, and assume there exists an
$x\in\left\{  0,1\right\}  ^{n}$\ such that $f\left(  x\right)  \neq g\left(
x\right)  $. \ Then letting $\rho_{f}$\ and $\rho_{g}$\ be the quantum
programs for $f$ and $g$ respectively, the fidelity $F\left(  \rho_{f}%
,\rho_{g}\right)  $\ must be at most $\epsilon$, for some $\epsilon$\ bounded
away from $1$ by a constant. \ (Otherwise $\rho_{f}$ and $\rho_{g}$\ would
lead to the same answers on $x$ with $1-o\left(  1\right)  $\ probability.)
\ This implies that $F(\rho_{f}^{\otimes k},\rho_{g}^{\otimes k})\leq
\epsilon^{k}$. \ So if we choose $k$ sufficiently large (say, more than $2m$),
then the set of states $\left\{  \rho_{f}^{\otimes k}\right\}  _{f\in
\mathcal{F}}$ is extremely close to an orthonormal basis. \ Thus, as in the
algorithm of Ettinger, H\o yer, and Knill \cite{ehk} for the nonabelian Hidden
Subgroup Problem, there must be a measurement of $\rho_{f}^{\otimes k}%
$\ (possibly exponentially hard to implement) that outputs $f$ with high probability.
\end{proof}

\section{Quantum Money\label{MONEYINT}}

We now consider the problem of developing public-key quantum money schemes.
\ First, in Section \ref{MONEYSCHEME}, we propose an explicit candidate scheme
for public-key quantum money, based on random stabilizer states. \ Then, in
Section \ref{MONEYORACLE}, we use the Complexity-Theoretic No-Cloning Theorem
to construct a quantum oracle relative to which public-key quantum money
schemes exist.

\subsection{The Random Stabilizer Scheme\label{MONEYSCHEME}}

Recall that a \textit{stabilizer state} is a pure state that can be obtained
by starting from $\left\vert 0\right\rangle ^{\otimes n}$\ and then applying
controlled-NOT, Hadamard, and $\pi/4$-phase gates, while a \textit{stabilizer
measurement} is a measurement that can be performed using those gates together
with computational basis measurements. \ (See Aaronson and Gottesman
\cite{ag}\ for details.) \ Given a security parameter $n$, let $\mathcal{D}%
_{n}$\ be the uniform distribution over all $n$-qubit stabilizer states.
\ Also, let $m,\ell,\varepsilon$ be additional parameters such that
$n/\varepsilon\ll m\ll1/\varepsilon^{2}\ll\ell$.

To generate a banknote, first the bank prepares $\ell$ stabilizer states
$\left\vert C_{1}\right\rangle ,\ldots,\left\vert C_{\ell}\right\rangle $,
which are drawn independently from $\mathcal{D}_{n}$. \ (It is well-known that
any stabilizer state can be prepared in polynomial time.) \ The bank
temporarily remembers the classical descriptions of the $\left\vert
C_{i}\right\rangle $'s, though it can erase those descriptions once the
preparation procedure is finished. \ Next, for each $i\in\left[  \ell\right]
$, the bank generates $m$\ random stabilizer measurements $E_{i1}%
,\ldots,E_{im}$\ as follows. \ For each $j\in\left[  m\right]  $:

\begin{itemize}
\item With $1-\varepsilon$\ probability, $E_{ij}$\ is a tensor product of $n$
uniformly random Pauli operators, with a random phase. \ That is,
$E_{ij}=\left(  -1\right)  ^{b}P_{1}\otimes\cdots\otimes P_{n}$, where $b$ is
drawn uniformly from $\left\{  0,1\right\}  $, and each $P_{k}$\ is drawn
uniformly from $\left\{  I,\sigma_{x},\sigma_{y},\sigma_{z}\right\}  $.

\item With $\varepsilon$\ probability, $E_{ij}$\ is a random tensor product of
Pauli operators as above, except that we condition on the event that
$\left\vert C_{i}\right\rangle $\ is a $+1$\ eigenstate of $E_{ij}$\ (that is,
$E_{ij}\left\vert C_{i}\right\rangle =\left\vert C_{i}\right\rangle $).
\end{itemize}

We can represent these $\ell m$\ measurements by a table $\mathcal{E}=\left(
E_{ij}\right)  _{ij}$, using $\left(  2n+1\right)  \ell m$ classical bits.
\ Finally, the bank generates an ordinary, classical digital signature
$\operatorname*{sig}\left(  \mathcal{E}\right)  $\ of\ the table $\mathcal{E}%
$,\ to prove that it and it alone could have generated $\mathcal{E}$. \ The
bank then distributes $\left(  \left\vert C_{1}\right\rangle ,\ldots
,\left\vert C_{\ell}\right\rangle ,\mathcal{E},\operatorname*{sig}\left(
\mathcal{E}\right)  \right)  $\ as the quantum banknote.

To authenticate such a banknote, one does the following. \ First check that
$\operatorname*{sig}\left(  \mathcal{E}\right)  $\ is a valid digital
signature for $\mathcal{E}$. \ Next, for each $i\in\left[  \ell\right]  $,
choose an index $j\left(  i\right)  \in\left[  m\right]  $\ uniformly at
random. \ Let $M$ be the two-outcome measurement that applies $E_{1j\left(
1\right)  }$\ to $\left\vert C_{1}\right\rangle $, $E_{2j\left(  2\right)  }%
$\ to $\left\vert C_{2}\right\rangle $, and so on up to $\left\vert C_{\ell
}\right\rangle $, and that accepts if and only if the \textit{majority} of
these measurements return a $+1$\ outcome (corresponding to $\left\vert
C_{i}\right\rangle $\ being a $+1$\ eigenstate of $E_{ij\left(  i\right)  }$).
\ Then apply $M$ to $\left\vert C_{1}\right\rangle \otimes\cdots
\otimes\left\vert C_{\ell}\right\rangle $, accept if and only if $M$ accepts,
and finally apply uncompute to get rid of garbage.

By construction, each $\left\vert C_{i}\right\rangle $\ will be measured to be
in a $+1$\ eigenstate of $E_{ij\left(  i\right)  }$\ with independent
probability $\frac{1-\varepsilon}{2}+\varepsilon=1/2+\varepsilon/2$. \ So by a
Chernoff bound, the probability that $M$ rejects is bounded away from $1$.
\ Indeed, we can make the probability that $M$ rejects exponentially small, by
simply taking $\ell$\ to be sufficiently larger than $1/\varepsilon^{2}$. \ By
Lemma \ref{goodasnew}, this implies that when we uncompute, we recover a state
that is exponentially close to $\left\vert C_{1}\right\rangle \otimes
\cdots\otimes\left\vert C_{\ell}\right\rangle $\ in trace distance---which in
turn implies that we can reuse the quantum banknote an exponential number of times.

On the other hand, we conjecture the following:

\begin{conjecture}
\label{moneyconj}Given $\left(  \left\vert C_{1}\right\rangle ,\ldots
,\left\vert C_{\ell}\right\rangle ,\mathcal{E},s\right)  $, it is
computationally infeasible not only to recover classical descriptions of the
states $\left\vert C_{1}\right\rangle ,\ldots,\left\vert C_{\ell}\right\rangle
$, but even to prepare additional copies of these states---or for that matter,
of any states that are accepted by the authentication procedure with
non-negligible probability.
\end{conjecture}

The intuition behind Conjecture \ref{moneyconj}\ is this: recovering classical
descriptions of $\left\vert C_{1}\right\rangle ,\ldots,\left\vert C_{\ell
}\right\rangle $\ given $\mathcal{E}$\ can be seen as a random instance of the
noisy decoding problem for linear codes, which is known to be $\mathsf{NP}%
$-complete in the worst case (see Berlekamp et al.\ \cite{berlekamp}).
\ Furthermore, while it is conceivable that a counterfeiter could use her
knowledge of $\mathcal{E}$\ to\ copy the $\left\vert C_{i}\right\rangle $'s
without learning classical descriptions of them, we have not found an
efficient way to do this. \ Indeed, it seems possible that to a
polynomial-time quantum algorithm---even one with knowledge of $\mathcal{E}%
$---the $\left\vert C_{i}\right\rangle $'s are actually indistinguishable from
$n$-qubit maximally mixed states.

Note that the scheme is \textit{not} secure if $m\leq n/\varepsilon$---since
then finding an $n$-qubit stabilizer state $\left\vert C_{i}\right\rangle
$\ that is accepted by an $\varepsilon$\ fraction of the measurements
$E_{i1},\ldots,E_{im}$\ is a trivial problem, solvable by Gaussian
elimination. \ Likewise, the scheme is not secure if $\varepsilon$ is too
large (say, greater than$\ 1/\sqrt{m}$)---since then one can recover the
stabilizer group of $\left\vert C_{i}\right\rangle $, with high probability,
by listing all measurements in the set $\left\{  E_{i1},\ldots,E_{im}\right\}
$\ that commute with suspiciously more than half\ of the other measurements in
the set.\footnote{We thank Peter Shor for this observation.} \ Thus,
Conjecture \ref{moneyconj}\ can only hold for suitable parameter ranges.

\subsection{Oracle Result\label{MONEYORACLE}}

If we allow ourselves the liberty of a quantum oracle, then we can prove the following.

\begin{theorem}
\label{moneythm}There exists a quantum oracle $U$ relative to which a
public-key quantum money scheme exists. \ (Here all parties---the bank,
authenticators, and counterfeiters---have the same access to $U$; no party has
\textquotedblleft inside information\textquotedblright\ about $U$ that is not
available to others.)
\end{theorem}

By \textquotedblleft quantum oracle,\textquotedblright\ we simply mean a
unitary transformation $U$ that can be applied in a black-box fashion. \ (We
may assume controlled-$U$ and $U^{-1}$\ are also available; this does not
particularly affect our results.) \ Quantum oracles were first studied in a
complexity-theoretic context by Aaronson and Kuperberg \cite{ak}, where they
were used to exhibit an oracle separation between the classes $\mathsf{QMA}$
and $\mathsf{QCMA}$.

In the proof of Theorem \ref{moneythm}, the oracle $U$\ does basically what
one would expect. \ Firstly, for each possible \textquotedblleft secret
key\textquotedblright\ $s\in\left\{  0,1\right\}  ^{n}$ that could be chosen
by the bank, the oracle maps the state $\left\vert 0\right\rangle \left\vert
s\right\rangle $ to $\left\vert 0\right\rangle \left\vert s\right\rangle
\left\vert e_{s}\right\rangle \left\vert \psi_{s}\right\rangle $,\ where
$e_{s}$\ is a classical \textquotedblleft public key\textquotedblright\ chosen
uniformly at random from $\left\{  0,1\right\}  ^{3n}$, and $\left\vert
\psi_{s}\right\rangle $\ is an $n$-qubit pure state chosen uniformly at random
under the Haar measure. \ (Of course, after being chosen at random, $e_{s}%
$\ and $\left\vert \psi_{s}\right\rangle $\ are then fixed for all time by the
oracle. \ Notice that with overwhelming probability, there is no pair
$s,s^{\prime}$\ such that $e_{s}=e_{s^{\prime}}$. \ Also, here and throughout
we omit ancilla qubits set to $\left\vert 0\cdots0\right\rangle $, when they
are part of the input to $U$.)

Secondly, for each $s\in\left\{  0,1\right\}  ^{n}$, the oracle maps the state
$\left\vert 1\right\rangle \left\vert e_{s}\right\rangle \left\vert \psi
_{s}\right\rangle $ to $\left\vert 1\right\rangle \left\vert e_{s}%
\right\rangle \left\vert \psi_{s}\right\rangle \left\vert 1\right\rangle $.
\ On the other hand, it maps $\left\vert 1\right\rangle \left\vert
e_{s}\right\rangle \left\vert \phi\right\rangle $\ to $\left\vert
1\right\rangle \left\vert e_{s}\right\rangle \left\vert \phi\right\rangle
\left\vert 0\right\rangle $\ if $\left\vert \phi\right\rangle $\ is orthogonal
to $\left\vert \psi_{s}\right\rangle $, and $\left\vert 1\right\rangle
\left\vert e\right\rangle \left\vert \phi\right\rangle $\ to $\left\vert
1\right\rangle \left\vert e\right\rangle \left\vert \phi\right\rangle
\left\vert 0\right\rangle $\ if $e\neq e_{s}$\ for every $s$.

By feeding $U$ inputs of the form $\left\vert 0\right\rangle \left\vert
s\right\rangle $, the bank can prepare and distribute an unlimited number of
banknotes $\left\vert e_{s}\right\rangle \left\vert \psi_{s}\right\rangle
$.\ \ By feeding $U$ inputs of the form $\left\vert 1\right\rangle \left\vert
e_{s}\right\rangle \left\vert \psi_{s}\right\rangle $, buyers and sellers can
then authenticate these banknotes. \ Furthermore, by the optimality of
Grover's algorithm \cite{bbbv}, it is clear that any would-be counterfeiter
needs $\Omega\left(  2^{n/2}\right)  $\ queries to $U$ to find the secret key
$s$, even if given the public key $e_{s}$.

So the real question is this: given $e_{s}$ \textit{together with} $\left\vert
\psi_{s}\right\rangle ^{\otimes k}$ for some $k=\operatorname*{poly}\left(
n\right)  $, can a counterfeiter, by making $\operatorname*{poly}\left(
n\right)  $\ queries to $U$, prepare a state that has non-negligible
overlap\ with $\left\vert \psi_{s}\right\rangle ^{\otimes k+1}$? \ We observe
that a negative answer follows more-or-less immediately from Theorem
\ref{ctnc0}, the Complexity-Theoretic No-Cloning Theorem.

\section{Quantum Copy-Protection\label{CPINT}}

Having summarized our results about quantum money, we now move on to the
related problem of copy-protecting quantum software.

\subsection{Two Schemes for Copy-Protecting Point Functions\label{POINTFN}}

Recall that a point function $f_{s}:\left\{  0,1\right\}  ^{n}\rightarrow
\left\{  0,1\right\}  $ has the form%
\[
f_{s}\left(  x\right)  =\left\{
\begin{array}
[c]{cc}%
1 & \text{if }x=s\\
0 & \text{otherwise}%
\end{array}
\right.
\]
In this section we propose two explicit schemes for quantumly copy-protecting
the family $\left\{  f_{s}\right\}  _{s\in\left\{  0,1\right\}  ^{n}}$\ of
point functions.

The first scheme, which we are grateful to Adam Smith for suggesting, uses a
pseudorandom generator $g:\left\{  0,1\right\}  ^{n}\rightarrow\left\{
0,1\right\}  ^{p\left(  n\right)  }$, where $p$ is some reasonably large
polynomial (say $n^{3}$). \ Given the secret key $s$, the software vendor
first computes $g\left(  s\right)  $, then reinterprets $g\left(  s\right)  $
as a description of a quantum circuit $U_{g\left(  s\right)  }$\ over some
universal basis of gates, which acts on $m$ qubits for some $m\ll n$. \ The
vendor then outputs $\left\vert \psi_{s}\right\rangle =U_{g\left(  s\right)
}\left\vert 0\right\rangle ^{\otimes m}$\ as its quantum program for $f_{s}$.
\ Given $\left\vert \psi_{s}\right\rangle $, the customer can efficiently
compute $f_{s}\left(  x\right)  $ for any $x$, by measuring the state
$U_{g\left(  x\right)  }^{-1}\left\vert \psi_{s}\right\rangle $\ in the
standard basis and then checking whether the outcome is $\left\vert
0\right\rangle ^{\otimes m}$.

Harrow and Low \cite{harrowlow} have recently shown that random quantum
circuits are approximate unitary $2$-designs. \ From this it follows that if
$x\neq s$, then $\left\vert \left\langle \psi_{x}|\psi_{s}\right\rangle
\right\vert $\ must be exponentially small with overwhelming probability,
unless $g$ is insecure against $2^{m}$-time classical adversaries. \ It is
also clear that $s$\ cannot be learned by a polynomial-time measurement on
$\left\vert \psi_{s}\right\rangle ^{\otimes k}$\ for any
$k=\operatorname*{poly}\left(  n\right)  $, unless $g$ is insecure against
polynomial-time quantum adversaries. \ However, the key conjecture is the following:

\begin{conjecture}
Given $\left\vert \psi_{s}\right\rangle ^{\otimes k}$, no polynomial-time
quantum algorithm can prepare a $\left(  k+1\right)  ^{\text{st}}$\ copy of
$\left\vert \psi_{s}\right\rangle $, or indeed, any other state from which
$f_{s}$\ can be efficiently computed.
\end{conjecture}

Our second candidate scheme is based on the Hidden Subgroup Problem over the
symmetric group. \ Given the secret key $s\in\left\{  0,1\right\}  ^{n}$, the
software vendor first encodes $s$, in some canonical way, as a permutation
$\tau_{s}\in S_{n}$\ such that $\tau_{s}^{2}=e$\ is the identity. \ The vendor
then prepares a state of the form%
\[
\left\vert \psi_{s}\right\rangle =\frac{\left\vert \sigma_{1}\right\rangle
+\left\vert \sigma_{1}\tau_{s}\right\rangle }{\sqrt{2}}\otimes\cdots
\otimes\frac{\left\vert \sigma_{k}\right\rangle +\left\vert \sigma_{k}\tau
_{s}\right\rangle }{\sqrt{2}},
\]
where $\sigma_{1},\ldots,\sigma_{k}$\ are permutations chosen uniformly at
random from $S_{n}$. \ Finally, the vendor distributes $\left\vert \psi
_{s}\right\rangle $ as the (amplified) quantum program for $f_{s}$. \ Given
$\left\vert \psi_{s}\right\rangle $, the customer can compute $f_{s}\left(
x\right)  $ for any $x$, by mapping each state $\frac{\left\vert \sigma
_{i}\right\rangle +\left\vert \sigma_{i}\tau_{s}\right\rangle }{\sqrt{2}}$\ to%
\[
\frac{1}{2}\left[  \left\vert 0\right\rangle \left(  \left\vert \sigma
_{i}\right\rangle +\left\vert \sigma_{i}\tau_{s}\right\rangle \right)
+\left\vert 1\right\rangle \left(  \left\vert \sigma_{i}\tau_{x}\right\rangle
+\left\vert \sigma_{i}\tau_{s}\tau_{x}\right\rangle \right)  \right]  ,
\]
then Hadamarding the first qubit and measuring it in the standard basis. \ If
$\tau_{x}=\tau_{s}$, then outcome $\left\vert 0\right\rangle $\ will be
obtained with certainty, while if $\tau_{x}\neq\tau_{s}$,\ then outcome
$\left\vert 1\right\rangle $\ will be obtained with\ probability $1/2$.

On the other hand, recovering $\tau_{s}$\ given $\left\vert \psi
_{s}\right\rangle $ is clearly at least as hard as the Hidden Subgroup Problem
(HSP) over the symmetric group, at least for subgroups $H\leq S_{n}$\ of order
$2$. \ Solving this special case of HSP would lead to a polynomial-time
quantum algorithm for the Rigid Graph Isomorphism problem. \ Furthermore,
Hallgren et al.\ \cite{hmrrs} have shown that any quantum algorithm for
recovering $\tau_{s}$\ would require entangled measurements on $\Omega\left(
n\log n\right)  $\ coset states; such an algorithm seems beyond present-day techniques.

Again, though, the conjecture we need is a stronger one:

\begin{conjecture}
Given $\left\vert \psi_{s}\right\rangle $, no polynomial-time quantum
algorithm can prepare an \textit{additional} coset state $\frac{\left\vert
\sigma_{k+1}\right\rangle +\left\vert \sigma_{k+1}\tau_{s}\right\rangle
}{\sqrt{2}}$, or indeed, any other state from which $f_{s}$\ can be
efficiently computed.
\end{conjecture}

The copying problem clearly reduces to HSP, but we do not know of a reduction
in the other direction.

\subsection{Oracle Result\label{CPORACLE}}

Our main result about quantum copy-protection is the following:

\begin{theorem}
\label{learnthm}There exists a quantum oracle $U$, relative to which any
family $\mathcal{F}$\ of efficiently computable functions that is not
quantumly learnable can be quantumly copy-protected (with security $\delta$,
against pirates mapping $k$ programs to $k+r$\ with $\left(  1-2\delta\right)
r>k$).
\end{theorem}

By a function family $\mathcal{F}$\ being \textquotedblleft quantumly
learnable,\textquotedblright\ we mean that given quantum oracle access to any
function $f\in\mathcal{F}$, one can in polynomial time prepare a state
$\left\vert \varphi_{f}\right\rangle $ from which $f$\ can then be computed in
polynomial time without further help from the oracle. \ As discussed before,
it is clear that no learnable family of functions can be copy-protected.
\ Theorem \ref{learnthm} says that this is the \textit{only} relativizing
obstruction to quantum copy-protection.

In the remainder of this section, we explain the essential steps in the proof
of Theorem \ref{learnthm}, in the special case where we only need to protect
against pirating algorithms that \textit{more than double} the number of
programs.

The oracle $U$ does the following. \ Given as input a state of the form
$\left\vert 0\right\rangle \left\vert d_{f}\right\rangle $, where $d_{f}$ is a
classical description of a Boolean function $f\in\mathcal{F}$, the oracle
outputs $\left\vert 0\right\rangle \left\vert d_{f}\right\rangle \left\vert
K_{f}\right\rangle \left\vert \psi_{f}\right\rangle $, where $K_{f}$\ is a
random classical codeword specifying $f$, and $\left\vert \psi_{f}%
\right\rangle $\ is a $2n$-qubit \textquotedblleft code
state\textquotedblright\ chosen uniformly at random under the Haar measure for
each $f$. \ Given as input a state of the form $\left\vert 1\right\rangle
\left\vert K_{f}\right\rangle \left\vert \psi_{f}\right\rangle \left\vert
x\right\rangle $, for some $x\in\left\{  0,1\right\}  ^{n}$, the oracle
outputs $\left\vert 1\right\rangle \left\vert K_{f}\right\rangle \left\vert
\psi_{f}\right\rangle \left\vert x\right\rangle \left\vert f\left(  x\right)
\right\rangle $. \ Given as input a state of the form $\left\vert
1\right\rangle \left\vert K_{f}\right\rangle \left\vert \phi\right\rangle
\left\vert x\right\rangle $, for any $\left\vert \phi\right\rangle
$\ orthogonal to $\left\vert \psi_{f}\right\rangle $, the oracle outputs
$\left\vert 1\right\rangle \left\vert K_{f}\right\rangle \left\vert
\phi\right\rangle \left\vert x\right\rangle \left\vert 0\right\rangle $.

It is clear that, for any function $f\in\mathcal{F}$, the software vendor can
create and distribute states of the form $\left\vert K_{f}\right\rangle
\left\vert \psi_{f}\right\rangle $, from which $f\left(  x\right)  $\ can be
efficiently computed for any input $x$. \ Furthermore, by the optimality of
Grover's algorithm, a software pirate has little hope of using the oracle $U$
to find $d_{f}$, given only $K_{f}$. \ As in the quantum money case, the real
question is this: given the state $\left\vert K_{f}\right\rangle \left\vert
\psi_{f}\right\rangle ^{\otimes k}$ for some $k=\operatorname*{poly}\left(
n\right)  $, can a quantum pirate produce $\ell>k$\ programs for $f$\ using
only $\operatorname*{poly}\left(  n\right)  $\ queries to $U$?

The Complexity-Theoretic No-Cloning Theorem suggests that the answer should be
no. \ However, we now have to handle a new difficulty that did not arise in
the money case. \ The new difficulty is that for certain function families
$\mathcal{F}$---namely, the learnable families---we know that it \textit{is}
possible to pirate $\left\vert \psi_{f}\right\rangle $ efficiently, by using
$\left\vert \psi_{f}\right\rangle $\ to simulate an oracle for $f$, and then
learning a new quantum program for $f$ just from $f$'s input/output behavior.
\ Thus, our proof will need to show that learnability is the only obstacle to
copy-protection. \ Or taking the contrapositive, we need to construct a
\textit{simulator}, which takes as input a polynomial-time algorithm for
pirating $\left\vert \psi_{f}\right\rangle $, and converts it into a
polynomial-time algorithm that learns a quantum program for $f$ using only
oracle access to $f$ (and no oracle access to $U$).

How should the simulator work? \ For simplicity, let us restrict ourselves to
simulators that use the pirating algorithm as a black box in constructing the
learning algorithm. \ Intuitively, what the simulator ought to do is

\begin{enumerate}
\item[(1)] \textquotedblleft mock up\textquotedblright\ its own stand-in
$\left\vert K\right\rangle \left\vert \varphi\right\rangle ^{\otimes k}$\ for
the state $\left\vert K_{f}\right\rangle \left\vert \psi_{f}\right\rangle
^{\otimes k}$,

\item[(2)] run the pirating algorithm on $\left\vert K\right\rangle \left\vert
\varphi\right\rangle ^{\otimes k}$, using the simulator's own oracle access to
$f$ to simulate the pirating algorithm's oracle calls to $U$ on inputs of the
form $\left\vert 1\right\rangle \left\vert K_{f}\right\rangle \left\vert
\psi_{f}\right\rangle \left\vert x\right\rangle $, and then

\item[(3)] use the output of the pirating algorithm to get an oracle-free
quantum program for $f$.
\end{enumerate}

The idea behind step (3) is as follows:\ we know that at least \textit{some}
of the programs output by the pirating algorithm must not make essential use
of the oracle $U$. \ For the oracle can only be usefully accessed via the
\textquotedblleft pseudorandom\textquotedblright\ state $\left\vert
\varphi\right\rangle $---and by the Complexity-Theoretic No-Cloning Theorem,
the simulator cannot have produced any additional copies of $\left\vert
\varphi\right\rangle $.

However, already at step (1) of the above plan, we encounter a problem: in the
oracle world, the states $\left\vert \psi_{f}\right\rangle $ were chosen
uniformly at random under the Haar measure. \ In polynomial time, with no
oracle access, how does one \textquotedblleft mock up\textquotedblright\ a
$2n$-qubit state $\left\vert \varphi\right\rangle $\ such that $\left\vert
\varphi\right\rangle ^{\otimes k}$\ behaves indistinguishably from $k$\ copies
of a uniform random state? \ This is the question that we answer in the full
version using Theorem \ref{tdesthm}, which gives an explicit quantum
$t$-design for arbitrary $t=\operatorname*{poly}\left(  n\right)  $\ with the
properties we need.

Let us now explain how the pieces are put together. \ Assume that $\left(
1-2\delta\right)  r>k$. \ Suppose we are given a pirating algorithm that takes
$\left\vert \psi_{f}\right\rangle ^{\otimes k}$\ as input (for a given
$f\in\mathcal{F}$), makes $T$ queries to the quantum oracle $U$, and outputs
$k+r$\ possibly-entangled quantum programs $\sigma_{1}^{U},\ldots,\sigma
_{k+r}^{U}$ such that%
\[
\sum_{i=1}^{k+r}\Pr_{\left(  f,x\right)  \in\mathcal{D}}\left[  L^{U}\left(
\sigma_{i}^{U},x\right)  \text{ outputs }f\left(  x\right)  \right]  \geq
k+\left(  1-\delta\right)  r.
\]
From this pirating algorithm, we want to obtain a polynomial-time algorithm
that uses oracle access to $f$ to learn an (oracle-free) quantum program for
$f$. \ Here is how it works:

\begin{enumerate}
\item[(1)] The simulator chooses some $t=\operatorname*{poly}\left(
k,T,n\right)  $. \ It then chooses a $2n$-qubit state $\left\vert
\varphi\right\rangle $\ uniformly at random from a quantum $t$-design, in the
sense of Theorem \ref{tdesthm}. \ The simulator also chooses a random string
$K$.

\item[(2)] The simulator creates a \textit{simulated oracle}\ $\widetilde{U}$,
which maps $\left\vert 1\right\rangle \left\vert K\right\rangle \left\vert
\varphi\right\rangle \left\vert x\right\rangle $\ to $\left\vert
1\right\rangle \left\vert K\right\rangle \left\vert \varphi\right\rangle
\left\vert x\right\rangle \left\vert f\left(  x\right)  \right\rangle $\ and
$\left\vert 1\right\rangle \left\vert K\right\rangle \left\vert \phi
\right\rangle \left\vert x\right\rangle $\ to $\left\vert 1\right\rangle
\left\vert K\right\rangle \left\vert \phi\right\rangle \left\vert
x\right\rangle \left\vert 0\right\rangle $ for every $\left\vert
\phi\right\rangle $\ orthogonal to $\left\vert \varphi\right\rangle $. \ (As a
technicality, $\widetilde{U}$\ does nothing on inputs of the form $\left\vert
0\right\rangle \left\vert d_{f}\right\rangle $.\footnote{For simplicity, we
are assuming it is exponentially hard for anyone but the software vendor to
guess the classical description $d_{f}$ for even a \textit{single} function
$f\in\mathcal{F}$---in which case, no one but the software vendor ever has
anything to gain by querying $U$ on inputs of the form $\left\vert
0\right\rangle \left\vert d\right\rangle $. \ With slightly more work, one can
remove this assumption, and even assume $d_{f}$ has some standard form such as
a description of a circuit for $f$.}) \ Note that $\widetilde{U}$\ can be
implemented in polynomial time, using the simulator's oracle access to $f$.

\item[(3)] The simulator runs the pirating algorithm, except with $\left\vert
K\right\rangle \left\vert \varphi\right\rangle ^{\otimes k}$ in place of
$\left\vert K_{f}\right\rangle \left\vert \psi_{f}\right\rangle ^{\otimes k}$
for the input, and queries to $\widetilde{U}$\ in place of queries to $U$.
\ The simulator outputs $\left\vert \Phi\right\rangle $, the output of the
pirating algorithm, as its candidate for an oracle-free quantum program for
$f$.

\item[(4)] Let $\sigma_{1},\ldots,\sigma_{k+r}$ be the (possibly-entangled)
registers of $\left\vert \Phi\right\rangle $\ corresponding to the
$k+r$\ pirated programs. \ Then given an input $x\in\left\{  0,1\right\}
^{n}$ and freeloading algorithm $L$, one computes $f\left(  x\right)  $\ as
follows. \ Choose $i\in\left[  k+r\right]  $\ uniformly at random; then run
$L^{\widetilde{U}}\left(  \sigma_{i},x\right)  $\ with $\widetilde{U}%
$\ replaced by the identity transformation, and return $L$'s output as the
guess\ for $f\left(  x\right)  $.
\end{enumerate}

We claim that step (4) outputs $f\left(  x\right)  $\ with probability
non-negligibly greater than $1/2$. \ Notice that one can amplify the success
probability by repeating steps (1)-(3) $t^{\prime}=\operatorname*{poly}\left(
n\right)  $\ times to obtain the state $\left\vert \Phi\right\rangle ^{\otimes
t^{\prime}}$, then repeating step (4) on each copy of $\left\vert
\Phi\right\rangle $\ and outputting the majority answer.

The argument goes as follows. \ By Theorem \ref{ctnc0} (the
Complexity-Theoretic No-Cloning Theorem), it is impossible to use the original
pirating algorithm to produce $k+1$\ copies of the Haar-random state
$\left\vert \psi_{f}\right\rangle $. \ Indeed, there cannot even be a single
input $x\in\left\{  0,1\right\}  ^{n}$ such that given $x$, one can use the
output of the pirating algorithm (together with $\operatorname*{poly}\left(
n\right)  $\ additional queries to $U$) to prepare $k+1$\ copies of
$\left\vert \psi_{f}\right\rangle $. \ For then, by simply \textit{guessing}
$x$ and then using amplitude amplification, one could prepare $k+1$\ copies of
$\left\vert \psi_{f}\right\rangle $\ using only $O\left(  \sqrt{2^{n}%
}\operatorname*{poly}\left(  n\right)  \right)  $\ queries to $U$, whereas
Theorem \ref{ctnc0}\ implies that $\Omega\left(  2^{n}/\operatorname*{poly}%
\left(  n\right)  \right)  $\ queries are needed. \ (This is why we stipulated
that $\left\vert \psi_{f}\right\rangle $\ has $2n$ qubits rather than $n$.)

By Theorem \ref{tdesthm}, it follows that the output $\left\vert
\Phi\right\rangle $ cannot be used to prepare $k+1$\ copies of $\left\vert
\varphi\right\rangle $\ in the simulated case either---for otherwise, we would
be able to distinguish the real case from the simulated one.

As a consequence, when we run $L^{\widetilde{U}}\left(  \sigma_{i},x\right)
$\ for each $i\in\left[  k+r\right]  $, at least $r$\ of the $k+r$%
\ invocations must be unaffected when $\widetilde{U}$\ is replaced by the
identity transformation. \ For if an invocation \textit{is} affected, then by
the BBBV lower bound \cite{bbbv}, it must at some point have fed
$\widetilde{U}$\ an input state that has $\Omega\left(  1/\operatorname*{poly}%
\left(  n\right)  \right)  $\ fidelity with some state of the form $\left\vert
1\right\rangle \left\vert K\right\rangle \left\vert \varphi\right\rangle
\sum_{y}\alpha_{y}\left\vert y\right\rangle $. \ For those are the only states
on which $\widetilde{U}$\ behaves differently from the identity
transformation. \ Thus, we can prepare a \textquotedblleft clock
state\textquotedblright\ of the form $\frac{1}{\sqrt{T}}\sum_{t=1}%
^{T}\left\vert t\right\rangle $, and use that state to determine how many
steps $t$ of $L$\ to apply to $\sigma_{i}$. \ We can then apply
$\operatorname*{poly}\left(  n\right)  $\ steps of amplitude amplification to
the joint state of the clock register and the $\sigma_{i}$\ register,
searching for a marked item of the form $\left\vert t\right\rangle
\otimes\left\vert \varphi\right\rangle $\ for any $t$. \ This will produce, in
the $\sigma_{i}$\ register, a state having $1-\varepsilon$\ fidelity with
$\left\vert \varphi\right\rangle $. \ But we already decided that this can be
done for at most $k$\ registers.

In summary, the expression%
\[
\sum_{i=1}^{k+r}\Pr_{\left(  f,x\right)  \in\mathcal{D}}\left[  L^{\widetilde
{U}}\left(  \sigma_{i},x\right)  \text{ outputs }f\left(  x\right)  \right]
\]
can decrease by at most (say) $k+2^{-n}$\ when $\widetilde{U}$\ is replaced by
$I$. \ Since $\left(  1-\delta\right)  r>k+\delta r$, this means the sum is at
least%
\begin{align*}
k+\left(  1-\delta\right)  r-\left(  k+2^{-n}\right)   &  =\left(
1-\delta\right)  r-2^{-n}\\
&  \geq k+\delta r-2^{-n}.
\end{align*}
So for\ $i\in\left[  k+r\right]  $\ chosen randomly, $L^{I}\left(  \sigma
_{i},x\right)  $ outputs the correct value of $f\left(  x\right)  $\ with
probability bounded above $1/2$, as claimed.

\section{Open Problems\label{OPEN}}

Can we find more explicit candidate schemes for public-key quantum money---and
better yet, prove such a scheme secure under a standard assumption?

Can we find candidate schemes for quantumly copy-protecting richer families of
functions than just point functions? \ What about trapdoor inversion functions?

Can we prove a scheme\ for copy-protecting point functions (such as those in
Section \ref{POINTFN}) secure under a standard assumption?

Can we improve Theorem \ref{learnthm} to remove the restriction on $r$?

Can a public-key (or at least query-secure) quantum money scheme exist, that
does not require multi-qubit entanglement in the banknotes? \ What about a
scheme for copy-protecting point functions that does not require multi-qubit
entanglement in the programs?

Can we show that public-key quantum money schemes exist relative to a
\textit{classical} oracle, rather than a quantum oracle? \ What about
nontrivial copy-protection schemes?

Is there a way to amplify a quantum program\ \textquotedblleft
unsplittably\textquotedblright---i.e., such that one cannot efficiently
decompose the amplified program into two somewhat-less-amplified programs, as
$\rho^{\otimes k}$\ can be decomposed into $\rho^{\otimes k/2}\otimes
\rho^{\otimes k/2}$?

Can we improve the parameters of the Complexity-Theoretic No-Cloning Theorem?

Can the Goldreich-Goldwasser-Micali reduction \cite{ggm} from PRGs to PRFs be
adapted to work in the presence of quantum adversaries?

Can we find a function family which is quantumly obfuscatable, but is not (or
is not known to be) classically obfuscatable?

Can we give constructions for unclonable quantum ID cards or quantum proofs?
\ How do these functionalities relate to money and copy-protection?

What can we say about \textit{information-theoretically secure} quantum
copy-protection, in the regime where the number of copies of the quantum
program is assumed to be small?

\section{Acknowledgments}

I am grateful to Michele Mosca and Douglas Stebila for making the connection
between the Complexity-Theoretic No-Cloning Theorem and quantum money; Adam
Smith for pointing me to Wiesner's article and suggesting the copy-protection
scheme based on random quantum circuits; Andris Ambainis and Luca Trevisan for
discussions about $t$-designs; Zeev Dvir for Theorem \ref{quantumprf}; the
anonymous reviewers for their comments; and Anne Broadbent, Ran Canetti, Aram
Harrow, Avinatan Hassidim, Peter Shor, and Salil Vadhan for helpful conversations.

{\small
\bibliographystyle{plain}
\bibliography{thesis}
}

\end{document}